\documentclass[aps,prl,reprint,superscriptaddress,amsmath,amssymb,nofootinbib]{revtex4-2}

\usepackage{amsthm}
\usepackage{mathtools}
\usepackage{enumitem}
\usepackage{tikz}
\usepackage[colorlinks=true,linkcolor=blue,citecolor=blue,urlcolor=blue]{hyperref}

\newcommand{\Mh}{\widehat{M}}
\newcommand{\Iop}{\widetilde{I}}
\newcommand{\R}{\mathbb{R}}
\newcommand{\Q}{\mathbb{Q}}
\newcommand{\Z}{\mathbb{Z}}
\newcommand{\T}{\mathbb{T}}
\newcommand{\bM}{\mathbf{M}}
\newcommand{\cR}{\mathcal{R}}
\newcommand{\cT}{\mathcal{T}}
\newcommand{\cK}{\mathcal{K}}
\DeclareMathOperator{\conv}{conv}
\DeclareMathOperator{\relint}{relint}

\newtheoremstyle{smplain}{4pt}{4pt}{\itshape}{}{\bfseries}{.}{0.5em}{}
\newtheoremstyle{smdef}{4pt}{4pt}{}{}{\bfseries}{.}{0.5em}{}
\theoremstyle{smplain}
\newtheorem{theorem}{Theorem}
\newtheorem{lemma}[theorem]{Lemma}
\newtheorem{proposition}[theorem]{Proposition}
\newtheorem{corollary}[theorem]{Corollary}
\theoremstyle{smdef}
\newtheorem{definition}[theorem]{Definition}

\begin{document}

\title{Arithmetic of Bohr Frequencies Governs Uniform Mixing\\ in Randomly Timed Quantum Spin Chains}

\author{Musung Kang}
\email{musung098@snu.ac.kr}
\affiliation{Seoul National University, Seoul, Republic of Korea}

\begin{abstract}
Randomizing the readout time averages the mode interference in a quantum walk.
We ask when this makes the site populations of a uniformly coupled $XY$ chain
exactly uniform for every initial site. For $N\ge2$ spins, this is possible
when $N+1$ is a power of two, a prime, or twice a prime.
It is impossible when $15\mid(N+1)$, $21\mid(N+1)$, or
$6\mid(N+1)$ with $N\ge11$.
Uniformity uniquely fixes the averaged cosine coherences according to the
mirror parities of the modes. Equal Bohr frequencies can impose incompatible
coherence values. We classify these parity collisions using vanishing sums
of roots of unity and prove existence in the positive cases through a
phase distribution on a torus.
\end{abstract}

\maketitle

\textit{Introduction.}
Uniformly coupled spin chains are simple quantum wires \cite{Bose03,Christandl04}.
An excitation injected at one site spreads coherently.
The transition probabilities at time $t$ form a doubly stochastic matrix $M(t)$.
Much research has focused on concentrating the excitation at a target site.
Perfect state transfer between the ends is impossible in uniform chains
of more than three spins \cite{Christandl04}.
Pretty good state transfer (PGST) allows fidelity arbitrarily close to one.
Between the ends of a nontrivial chain, it occurs exactly when $N+1$ is a prime,
twice a prime, or a power of two \cite{GKSS12}.
Arithmetic thus controls quantum transport.

Here we ask when an excitation can instead be spread uniformly.
We require a single readout protocol to work for every initial site.
For a chain of three or more spins, the full transition matrix cannot
approach $J/N$ arbitrarily closely at a single time \cite{Mon26}.
We prove below that the only trees with this approximation property
are the one-vertex graph, the two-vertex path, and the four-vertex star.

Averaging over time changes the problem.
If the readout time $T$ has distribution $\mu$, the observed transition
matrix is $\Mh[\mu]=\mathbb E\,M(T)$.
The timing distribution may arise from clock jitter or be chosen deliberately.
Baptista, Coutinho and Marques \cite{BCM24} found uniformly mixing
distributions for chains of two, three and four spins.
They asked whether every chain admits such a distribution.
We establish existence and nonexistence for two infinite families of chain lengths.

\medskip
\noindent\textbf{Theorem.} \textit{Consider a uniformly coupled $XY$ chain of $N$ spins and put
$q=N+1\ge3$.}
\begin{enumerate}
\item[(i)] \textit{If $q$ is a power of two, a prime, or twice a prime, then there are finitely many
readout times $t_j$ with weights $\mu_j$, and also a probability density, for which the averaged
population of every site is exactly $1/N$, for every initial site.}
\item[(ii)] \textit{If $6\mid q$ with $q\ge12$, or $15\mid q$, or $21\mid q$, then no probability
distribution of readout times has this property.}
\end{enumerate}
\medskip
The two lists are disjoint.
The main steps of the proofs are collected in the End Matter, and complete proofs are given in the Supplemental Material \cite{SM}.
Here we explain the mechanisms.
Cases other than (i) and (ii) are not decided by the Theorem.

\textit{Model.}
The Hamiltonian
$\hat H=\frac12\sum_{u=1}^{N-1}
(\sigma^x_u\sigma^x_{u+1}+\sigma^y_u\sigma^y_{u+1})$
preserves the number of excitations.
On the single-excitation subspace, it acts as the adjacency matrix $A$
of the path $P_N$ \cite{Christandl04,GKSS12}.
Its eigenvalues and spectral idempotents are
\begin{equation}\label{eq:spectrum}
\theta_r=2\cos\frac{r\pi}{q},\quad E_r=v_rv_r^{\mathsf T},
\end{equation}
where $(v_r)_k=\sqrt{2/q}\sin(rk\pi/q)$ for $1\le r,k\le N$.
We use $U(t)=e^{itA}$ and $M(t)_{kl}=|U(t)_{kl}|^2$.
The opposite sign convention for time evolution gives the same $M(t)$.
Averaging the spectral expansion over $\mu$ gives \cite{BCM24}
\begin{equation}\label{eq:BCM}
\Mh[\mu]=\sum_rE_r\circ E_r+2\sum_{r<s}y_{rs}\,E_r\circ E_s,
\end{equation}
where $\circ$ denotes the entrywise product and
\begin{equation}\label{eq:coh}
y_{rs}=\int_\R\cos\bigl((\theta_r-\theta_s)t\bigr)\,d\mu(t).
\end{equation}
We call these averaged cosines the \emph{coherences}.
They are the only features of $\mu$ on which $\Mh[\mu]$ depends.
We call $\mu$ \emph{uniformly mixing} if $\Mh[\mu]=J/N$,
where $J$ is the all-ones matrix.

\textit{Forced coherences.} Equation \eqref{eq:BCM} is linear in the coherences, but it is not
obvious that uniformity determines them, since the matrices $E_r\circ E_s$ might be linearly dependent.
Our first result is that uniformity does determine them.

\medskip
\noindent\textbf{Lemma (forced coherences).} \textit{If $\Mh[\mu]=J/N$, then for all $1\le r<s\le N$}
\begin{equation}\label{eq:forced}
y_{rs}=\begin{cases}-1/N, & s-r \text{ even},\\ 
0, & s-r \text{ odd}.\end{cases}
\end{equation}
\medskip

The mode $v_r$ has parity $(-1)^{r-1}$ under the mirror reflection
$k\mapsto q-k$. Thus Eq.~\eqref{eq:forced} requires zero cosine coherence
between modes of opposite mirror parity.
Between distinct modes of equal parity, it requires the same negative
value $-1/N$, regardless of the Bohr frequency.

Two facts prove the lemma.
First, the values in Eq.~\eqref{eq:forced} give $J/N$.
The reflection matrix $\Iop_{kl}=\delta_{k+l,q}$ satisfies
$\sum_r(-1)^rE_r=-\Iop$.
The long-time Ces\`aro average is
\[
D_0:=\lim_{L\to\infty}\frac1L\int_0^L M(t)\,dt
=\sum_r E_r\circ E_r
=\frac{2J+I+\Iop}{2q}
\]
\cite{God13}.
Substitution of these identities into Eq.~\eqref{eq:BCM} gives the target.

Second, the coherence values are unique.
The product-to-sum identity gives
\begin{equation}\label{eq:rankone}
\begin{aligned}
E_r\circ E_s&=q^{-2}w_{rs}w_{rs}^{\mathsf T},\\
w_{rs}&=c_{s-r}-c_{r+s},
\end{aligned}
\end{equation}
where $(c_j)_k=\cos(jk\pi/q)$.
The pairs $(r,s)$ and $(q-s,q-r)$ have the same matrix and Bohr frequency.
Their coherences therefore agree.
After identifying these pairs, the matrices in Eq.~\eqref{eq:rankone}
are linearly independent.
For $0\le a\le q$ and $1\le j\le N$, the Gram identity is
\[
\langle c_a,c_j\rangle
=\frac q2[a=j]-[a\equiv j\pmod2],
\]
where brackets denote indicators.
It converts a vanishing linear combination of the rank-one matrices
into a vanishing principal submatrix of a weighted graph Laplacian.
Its entries force every edge weight to vanish (End Matter).

The algebraic target is constant on the identified pairs.
Comparing it with any uniformly mixing distribution therefore gives
Eq.~\eqref{eq:forced}.
Equivalently, uniform average mixing occurs precisely when the target
coherence vector belongs to the convex hull of the coherence curve
traced by a single readout time.

\textit{Frequency collisions.}
Equation \eqref{eq:forced} prescribes one number per pair of modes,
but a distribution of times only sees Bohr frequencies.
Suppose two pairs share a Bohr frequency,
\begin{equation}\label{eq:collision}
\theta_k-\theta_l=\theta_{k'}-\theta_{l'},\quad l-k\not\equiv l'-k'\pmod 2 .
\end{equation}
Then the common value $\int\cos((\theta_k-\theta_l)t)\,d\mu$ must equal both $-1/N$ and $0$, so no
uniformly mixing distribution exists. We call this a \emph{parity collision}.
For $N=11$, $\theta_1-\theta_5=\theta_3-\theta_6=\sqrt2$ with index differences $4$ and $3$ (Fig.~\ref{fig:collision}).
For $q=15$ one has $\theta_1-\theta_3=\theta_4-\theta_5$, and for $q=21$ one has
$\theta_3-\theta_6=\theta_7-\theta_9$. The first example belongs to a family of collisions for every
$q=6m$ with $m\ge2$, and the other two persist, after rescaling the indices, for all odd multiples of
$15$ and $21$ (End Matter). For $N=11$ the absence of
uniform mixing was first obtained by an explicit dual certificate \cite{Kang26}.

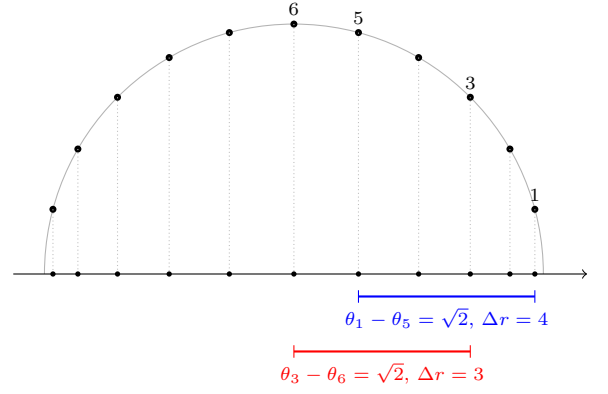
\begin{figure}[t]
\centering
\begin{tikzpicture}[scale=1.65]
  \draw[gray!60] (2,0) arc (0:180:2);
  \draw[->] (-2.25,0) -- (2.35,0);
  \foreach \r in {1,...,11} {
    \fill ({2*cos(15*\r)},{2*sin(15*\r)}) circle (0.8pt);
    \draw[gray!50,densely dotted] ({2*cos(15*\r)},{2*sin(15*\r)}) -- ({2*cos(15*\r)},0);
    \fill ({2*cos(15*\r)},0) circle (0.6pt);
  }
  \foreach \r in {1,3,5,6} {
    \node[above] at ({2*cos(15*\r)},{2*sin(15*\r)}) {\scriptsize $\r$};
  }
  \draw[thick,blue] ({2*cos(75)},-0.18) -- ({2*cos(15)},-0.18);
  \draw[blue] ({2*cos(75)},-0.13) -- ({2*cos(75)},-0.23);
  \draw[blue] ({2*cos(15)},-0.13) -- ({2*cos(15)},-0.23);
  \node[blue,below] at ({(2*cos(75)+2*cos(15))/2},-0.2)
     {\scriptsize $\theta_1-\theta_5=\sqrt2$, $\Delta r=4$};
  \draw[thick,red] (0,-0.62) -- ({2*cos(45)},-0.62);
  \draw[red] (0,-0.57) -- (0,-0.67);
  \draw[red] ({2*cos(45)},-0.57) -- ({2*cos(45)},-0.67);
  \node[red,below] at ({cos(45)},-0.64) {\scriptsize $\theta_3-\theta_6=\sqrt2$, $\Delta r=3$};
\end{tikzpicture}
\caption{Bohr-frequency collision in the chain of $N=11$ spins. The eigenvalues $\theta_r=2\cos(r\pi/12)$
are the projections of equally spaced points on a semicircle. The gaps $\theta_1-\theta_5$ and
$\theta_3-\theta_6$ coincide, but their index differences have opposite parity, so the forced coherences
\eqref{eq:forced} would have to be $-1/11$ and $0$ at the same frequency.}
\label{fig:collision}
\end{figure}

Conversely, every parity collision occurs at one of these moduli.
Writing $\pm\theta_j=\zeta^{x}+\zeta^{-x}$ with
$\zeta=e^{i\pi/q}$, a collision \eqref{eq:collision} becomes a vanishing sum of at most eight $2q$-th
roots of unity in which no two terms are antipodal.
The classification of minimal vanishing sums of few roots of unity, due to Conway and Jones \cite{CJ76} and extended by Poonen and Rubinstein \cite{PR98} in their count of the intersection points of the diagonals of a regular polygon, leaves only components
built from the relations of $3$rd, $5$th and $7$th roots of unity. Reading off which roots of unity can occur gives exactly the moduli of part (ii) of the Theorem.
The exceptional modulus $q=6$ has no parity collision.
Thus the arithmetic that counts crossing diagonals of a regular polygon also locates the spin chains in which uniform mixing is obstructed.

\textit{Construction.}
Suppose that the positive eigenvalues $\theta_r$ with $1\le r<q/2$
are linearly independent over $\Q$.
This occurs exactly when $q$ is a power of two, a prime, or twice a prime
(End Matter).
These are also the moduli for end-to-end PGST \cite{GKSS12}.
Kronecker's theorem \cite{HW08} makes the positive-mode phase flow
dense in its torus (Fig.~\ref{fig:torus}).
For even $q$, the zero-mode phase remains fixed at zero.
Negative-mode phases are the negatives of their positive partners.

For $q\ge4$, consider phases that vary linearly with the
mode index:
\[
\vartheta_a=(q/2-a)\xi
\]
for all modes. Each successive mode changes the phase by
the same amount, $\vartheta_{a+1}-\vartheta_a=-\xi$.
We call this linear progression a \emph{phase ramp}.
This choice respects the spectral pairing and fixes the
zero-mode phase when a zero mode is present.
Along the ramp,
\[
\cos(\vartheta_a-\vartheta_b)=\cos((b-a)\xi).
\]
The coherences therefore depend only on the index differences.

Take $\xi\in\R/4\pi\Z$ with density
\begin{equation}\label{eq:density}
\begin{aligned}
f_\lambda(\xi)&=1-\frac{2\lambda}{N}\sum_{j=1}^{g}\cos(2j\xi),\\
g&=\left\lfloor\frac{N-1}{2}\right\rfloor,
\end{aligned}
\end{equation}
with respect to normalized Haar measure.
For $\lambda>0$, this gives $\lambda$ times the target coherences.
The density is strictly positive when $\lambda<N/(2g)$,
and $N/(2g)>1$.

Haar measure on the phase torus has zero off-diagonal coherences.
Its full support places that vector in the relative interior
of the convex set of phase coherence vectors.
Choose $\lambda>1$ within the positivity range.
The target lies on the segment from the Haar vector to the ramp vector,
with positive weight on the Haar vector.
It therefore lies in the same relative interior.

The coherence vectors obtained from time distributions have the same
closed convex hull as those obtained from phase distributions.
Their relative interiors coincide, so the target is realized by
an actual time distribution.
Carath\'eodory's theorem gives a distribution supported on at most
$\lfloor q/2\rfloor^2+1$ times.
A second convexity argument gives an absolutely continuous distribution
(End Matter).
For $q=3$, a single time $t=\pi/4$ already suffices.
This proves part (i).
The argument proves existence but does not specify the readout times
or a time density.

\begin{figure}[t]
\centering
\begin{tikzpicture}[flow/.style={gray!65,line width=0.3pt}]
  \draw[flow] (0.0000,0.0000) -- (3.3000,1.2605);
  \draw[flow] (0.0000,1.2605) -- (3.3000,2.5210);
  \draw[flow] (0.0000,2.5210) -- (2.0395,3.3000);
  \draw[flow] (2.0395,0.0000) -- (3.3000,0.4815);
  \draw[flow] (0.0000,0.4815) -- (3.3000,1.7420);
  \draw[flow] (0.0000,1.7420) -- (3.3000,3.0024);
  \draw[flow] (0.0000,3.0024) -- (0.7790,3.3000);
  \draw[flow] (0.7790,0.0000) -- (3.3000,0.9629);
  \draw[flow] (0.0000,0.9629) -- (3.3000,2.2234);
  \draw[flow] (0.0000,2.2234) -- (2.8185,3.3000);
  \draw[flow] (2.8185,0.0000) -- (3.3000,0.1839);
  \draw[flow] (0.0000,0.1839) -- (3.3000,1.4444);
  \draw[flow] (0.0000,1.4444) -- (3.3000,2.7049);
  \draw[flow] (0.0000,2.7049) -- (1.5580,3.3000);
  \draw[flow] (1.5580,0.0000) -- (3.3000,0.6654);
  \draw[flow] (0.0000,0.6654) -- (3.3000,1.9259);
  \draw[flow] (0.0000,1.9259) -- (3.3000,3.1863);
  \draw[flow] (0.0000,3.1863) -- (0.2976,3.3000);
  \draw[flow] (0.2976,0.0000) -- (3.3000,1.1468);
  \draw[flow] (0.0000,1.1468) -- (3.3000,2.4073);
  \draw[flow] (0.0000,2.4073) -- (2.3371,3.3000);
  \draw[flow] (2.3371,0.0000) -- (3.3000,0.3678);
  \draw[flow] (0.0000,0.3678) -- (3.3000,1.6283);
  \draw[flow] (0.0000,1.6283) -- (3.3000,2.8888);
  \draw[flow] (0.0000,2.8888) -- (1.0766,3.3000);
  \draw[flow] (1.0766,0.0000) -- (3.3000,0.8493);
  \draw[flow] (0.0000,0.8493) -- (3.3000,2.1098);
  \draw[blue,line width=1.3pt] (0,0) -- (3.3,1.1);
  \draw[blue,line width=1.3pt] (0,1.1) -- (3.3,2.2);
  \draw[blue,line width=1.3pt] (0,2.2) -- (3.3,3.3);
  \draw[line width=0.6pt] (0,0) rectangle (3.3,3.3);
  \draw[->,line width=0.6pt] (1.5,0) -- (1.8,0);
  \draw[->,line width=0.6pt] (1.5,3.3) -- (1.8,3.3);
  \draw[->>,line width=0.6pt] (0,1.5) -- (0,1.8);
  \draw[->>,line width=0.6pt] (3.3,1.5) -- (3.3,1.8);
  \fill[red] (2.475,0.825) circle (2.2pt);
  \node[red,right] at (2.525,0.685) {\scriptsize PGST};
  \node[below] at (0,0) {\scriptsize $0$};
  \node[below] at (3.3,0) {\scriptsize $2\pi$};
  \node[left] at (0,3.3) {\scriptsize $2\pi$};
  \node[below] at (1.65,-0.12) {\small $\vartheta_1$};
  \node[left] at (-0.12,1.65) {\small $\vartheta_2$};
\end{tikzpicture}
\caption{Torus filling for $N=4$ ($q=5$), with opposite sides identified. Gray: the phase flow
$t\mapsto(t\theta_1,t\theta_2)$ modulo $2\pi$, a line of irrational slope $\theta_2/\theta_1$ that fills the
torus. Blue: the phase ramp $\vartheta_a=(\frac52-a)\xi$, a closed curve of slope $\frac13$ carrying the
slope density \eqref{eq:density}. Red: the PGST point $\xi=\pi$.}
\label{fig:torus}
\end{figure}
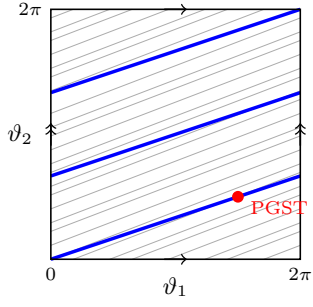

\textit{Relation to state transfer.}
At $\xi=\pi$, the ramp reaches the PGST phase vector:
$e^{i\vartheta_a}=i^q(-1)^a$.
For the constructive moduli, suitable times therefore give
$U(t_j)\to-i^q\Iop$.
Since $f_\lambda$ has period $\pi$, its phase law is invariant under
translation by this vector.
This invariance cancels the cosine coherences at odd index differences.

A parity collision obstructs the transfer phase.
Its integer vector
$\boldsymbol{\ell}=\mathbf{e}_k-\mathbf{e}_l
-\mathbf{e}_{k'}+\mathbf{e}_{l'}$,
where $\mathbf{e}_j$ denotes the $j$th standard basis vector,
satisfies
\[
\sum_a\ell_a\theta_a=0,\qquad \sum_a\ell_a=0.
\]
Every phase vector $\boldsymbol{\vartheta}$ in the
phase-flow closure satisfies
\[
\sum_a\ell_a\vartheta_a\equiv0\pmod{2\pi}.
\]
At the transfer phase, however,
$\sum_a\ell_a\vartheta_a\equiv\pi\pmod{2\pi}$.
Since $\sum_a\ell_a=0$, adding a global phase cannot
remove this obstruction. Thus the same collision also
excludes end-to-end PGST \cite{GKSS12}.
This restriction concerns the ends only.
The chain of $11$ spins has PGST between sites $2$ and $10$
\cite{CGvB17,vB19}, but admits no uniformly mixing distribution.

\textit{Beyond chains.}
On a bipartite graph, suppose that a sequence $U(t_j)$ converges
to a matrix $Q$ whose entries have modulus $N^{-1/2}$.
The entries of $Q$ are real within each sublattice and purely imaginary
across the two sublattices \cite{GMR17}.
Comparing entries in $AQ=QA$ then shows that all vertex degrees
have the same parity.

A tree with at least two vertices has a leaf.
All its degrees must therefore be odd.
If it has at least three vertices, a leaf count gives two leaves
with a common neighbor.
Their antisymmetric state is stationary.
The corresponding entries of $Q$ force $N\le4$.
It follows that the only trees with instantaneous or
arbitrarily good uniform mixing are $K_1$, $P_2$
and $K_{1,3}$. The star $K_{1,3}$ admits instantaneous
uniform mixing at $t_0=2\pi/\sqrt{27}$
\cite[Sec.~11]{GZ17}.

Cartesian products impose an additional constraint.
For $G\square H$, evolution factorizes as
$U(t)=U_G(t)\otimes U_H(t)$.
Uniform average mixing is equivalent to
\[
\int M_G(t)_{uv}M_H(t)_{u'v'}\,d\mu(t)
=\frac1{|V(G)|\,|V(H)|}
\]
for all entries.
This requires uniform average mixing in each factor and zero covariance
between every pair of factor transition probabilities.

For $G\square G$, the condition forces $M_G(t)=J/|V(G)|$
at almost every sampled time.
Thus a Cartesian square admits uniform average mixing exactly when
its factor admits instantaneous uniform mixing.
In particular, $P_{p-1}\square P_{p-1}$ admits no uniformly mixing
distribution for any prime $p\ge5$.
Each factor $P_{p-1}$ does admit one.

\textit{Outlook.}
Random timing turns uniform spreading into a problem about Bohr frequencies.
For paths, we obtain arithmetic existence criteria and a complete
classification of parity-collision obstructions.
The moduli outside the two families in the Theorem remain unresolved here.

A next step is to study other networks and Hamiltonians.
The forced-coherence argument requires an injective map from the relevant
moments to the averaged transition matrix.
For example, adding $\Delta\sum_{u=1}^{N-1}\sigma^z_u\sigma^z_{u+1}$
changes the single-excitation Hamiltonian to
\[
A+2\Delta(e_1e_1^{\mathsf T}+e_Ne_N^{\mathsf T})
\]
up to a scalar matrix.
For $N\ge3$ and nonzero $\Delta$, these boundary terms remove the anticommutation
with the sublattice signature used here.

Coupled photonic waveguides \cite{Perets08} and engineered spin or qubit
arrays \cite{Bloch08,Majer07} provide settings for controlled coherent
transport.
Our results identify infinite families of chain lengths for which
random readout can, or cannot, make the site distribution exactly uniform.

\begin{acknowledgments}
The author is supported by the National Research Foundation of Korea, Grant ID A0424-20260100, and by the
Hyunsong Scholarship.
\end{acknowledgments}

\section*{End Matter}

\textit{Target coherences.}
Put $D_0=\sum_rE_r\circ E_r$ and
$\mathcal E=\sum_{r<s,\ s-r\text{ even}}E_r\circ E_s$.
Reflection of the sine eigenvectors gives $\sum_r(-1)^rE_r=-\Iop$.
Consequently
\[
\sum_{r,s}E_r\circ E_s=I,\quad
\sum_{r,s}(-1)^{r+s}E_r\circ E_s=\Iop.
\]
Selecting equal-parity indices yields
$D_0+2\mathcal E=\frac12(I+\Iop)$.
Since $D_0=(2J+I+\Iop)/(2q)$ \cite{God13}, the target gives
\[
\Mh=D_0-\frac2N\mathcal E
=\frac qN D_0-\frac1{2N}(I+\Iop)
=\frac1N J.
\]

\textit{Uniqueness of the coherences.}
Equation~\eqref{eq:rankone} follows from
$2\sin\alpha\sin\beta=\cos(\alpha-\beta)-\cos(\alpha+\beta)$.
Represent each class $\{(k,l),(q-l,q-k)\}$ by a pair with $c=k+l\le q$.
Put $b=l-k$.
Then $1\le b<c\le q$ and $b\equiv c\pmod2$.
Distinct classes give distinct pairs $(b,c)$.

Geometric sums give
\[
\langle c_a,c_j\rangle
=\frac q2[a=j]-[a\equiv j\pmod2]
\]
for $0\le a\le q$ and $1\le j\le N$.
For $w_{bc}=c_b-c_c$, the parity terms cancel:
\[
\langle w_{bc},c_j\rangle=\frac q2([b=j]-[c=j]).
\]
Suppose that $\sum z_{bc}w_{bc}w_{bc}^{\mathsf T}=0$.
Pairing with $c_j$ and $c_{j'}$ makes the principal submatrix
on $\{1,\dots,N\}$ of a weighted graph Laplacian vanish.
The graph has vertices $\{1,\dots,q\}$ and edge weights $z_{bc}$.
Off-diagonal entries first give $z_{bc}=0$ for $c\le N$.
The remaining diagonal entries give $z_{bq}=0$.

Every time distribution has coherences constant on the identified classes.
The target in Eq.~\eqref{eq:forced} has this property as well.
Subtracting the target expansion from a uniformly mixing expansion
and applying independence proves that the two coherence vectors agree.

\textit{Collisions.}
If $6\mid q$ and $q\ge12$, put $h=q/3$.
Then $h$ is even and $h\ge4$.
For $1\le k\le h/2$,
\[
\begin{aligned}
\theta_k-\theta_{k+h}
&=4\sin\frac{(2k+h)\pi}{2q}\sin\frac{\pi}{6}\\
&=\theta_{h-k}-\theta_{q/2}.
\end{aligned}
\]
The index differences are $h$ and $h/2+k$.
Choose $k\in\{1,2\}$ so that the second is odd.

For odd $a$, the cases $q=15a$ and $q=21a$ give
\[
\begin{aligned}
&(\theta_a-\theta_{3a})-(\theta_{4a}-\theta_{5a})\\
&\quad=2\cos\frac{2\pi}{5}-2\cos\frac{\pi}{5}+1=0,\\
&(\theta_{3a}-\theta_{6a})-(\theta_{7a}-\theta_{9a})\\
&\quad=2\left(\cos\frac{\pi}{7}-\cos\frac{2\pi}{7}
+\cos\frac{3\pi}{7}\right)-1=0.
\end{aligned}
\]
The respective index differences are $(2a,a)$ and $(3a,2a)$.
Even cofactors already fall under the multiples of $6$.

Conversely, no two signed terms of a parity collision can cancel.
Such a cancellation would identify the two pairs up to reflection,
giving equal index differences.
In particular, at most one signed term is zero.
Discard that term if present.
Writing each remaining term as $\zeta^x+\zeta^{-x}$ with
$\zeta=e^{i\pi/q}$ gives a relation of weight $8$ or $6$.
An antipodal pair of roots would force either a discarded zero term
or two cancelling signed terms.
Thus no antipodal pair remains.

The classification of minimal vanishing sums of weight at most $8$
\cite{CJ76,PR98} leaves the types
\[
\begin{gathered}
R_3,\ R_5,\ R_7,\ (R_5{:}R_3),\ (R_5{:}2R_3),\\
(R_5{:}3R_3),\ (R_7{:}R_3).
\end{gathered}
\]
Here $R_p$ is the relation formed by all $p$th roots of unity.
In $(R_p{:}jR_3)$, each of $j$ distinct terms $\eta$ is replaced
by $-\eta\zeta_3$ and $-\eta\zeta_3^2$.

A rotated $R_p$ within the $2q$th roots requires $p\mid2q$.
Each mixed type contains two roots with ratio of order $6p$.
The possible component weights are $8$, $5+3$, $6$ and $3+3$.
The first three cases give $15\mid q$ or $21\mid q$.
The last gives $3\mid q$ and total weight $6$.
A zero signed term was therefore discarded, so one index is $q/2$.
Hence $q$ is even and $6\mid q$.
For $q=6$, the odd-difference gaps are
$\sqrt3-1,1,\sqrt3+1$, whereas the even-difference gaps are
$\sqrt3,2,2\sqrt3$.
These sets are disjoint.

\textit{Independence of the eigenvalues.}
Put $\cR^\times=\{r:1\le r<q/2\}$ and
$D=[\Q(\theta_1):\Q]=\frac12\varphi(2q)$.
Independence requires $|\cR^\times|\le D$.
For odd $q$, this gives $\varphi(q)\ge q-1$, so $q$ is prime.
For $q=2c$ with $c\ge3$ odd, it gives $\varphi(c)\ge c-1$,
so $c$ is prime.
For $q=2^ac$ with $a\ge2$ and $c$ odd, it gives
$2^{a-1}(c-\varphi(c))\le1$, so $c=1$.

Conversely, for $r\ge1$, each $\theta_r$ is a monic integer polynomial
of degree $r$ in $\theta_1$.
Thus $1,\theta_1,\dots,\theta_{D-1}$ form a basis of $\Q(\theta_1)$.
This proves independence when $q=2^a$ because $D=|\cR^\times|+1$.
For $q=p$ or $q=2p$ with $p$ an odd prime, $D=|\cR^\times|$.
Pairing conjugate $p$th roots gives
\[
1=-\sum_{j=1}^{(p-1)/2}\theta_{2qj/p}.
\]
The identity $\theta_{q-r}=-\theta_r$ puts this sum in the span
of the positive eigenvalues.
They therefore span the field and are independent.

\textit{Construction.}
Assume one of these independent-spectrum moduli and $q\ge4$.
Let $\cR=\{1,\dots,\lfloor q/2\rfloor\}$.
Kronecker's theorem \cite{HW08} identifies the phase-flow closure
with $\cT=\{\vartheta\in\T^{\cR}:\vartheta_{q/2}=0\}$.
The coordinate condition is omitted for odd $q$.
Extend phases to all modes by $\vartheta_{q-a}=-\vartheta_a$ and put
\[
Y(\vartheta)=(\cos(\vartheta_a-\vartheta_b))_{a<b}.
\]
Let $\cK$ be the convex hull of the time coherence curve.
The compact set $\cK_{\cT}=\conv Y(\cT)$ is its closure.
These two convex sets have the same relative interior \cite{Rock70}.

Haar measure on $\cT$ has coherence vector $0$.
Full support puts this vector in $\relint\cK_{\cT}$.
The ramp with density \eqref{eq:density} realizes $\lambda Y^\star$,
where $Y^\star$ is the target.
Choose $1<\lambda<N/(2g)$.
Then
\[
Y^\star=(1-\lambda^{-1})\,0+\lambda^{-1}(\lambda Y^\star)
\]
belongs to $\relint\cK_{\cT}$ and hence to $\cK$.

Write the phase moments as a symmetric pair $(C,S)$ indexed by $\cR$.
They satisfy $C_{rr}+S_{rr}=1$.
Their affine space has dimension at most $|\cR|^2$,
and $Y$ is an affine function of the pair.
Carath\'eodory's theorem therefore gives at most
$\lfloor q/2\rfloor^2+1$ readout times.

The coherence vectors of absolutely continuous time distributions
also form a convex set.
Uniform laws on short intervals approximate every point of the time curve.
Their convex set consequently has closure $\cK_{\cT}$ and the same
relative interior.
It therefore contains $Y^\star$.
For $q=3$, both the time $\pi/4$ and the uniform law on $[0,\pi]$ work.

\textit{Degree parity and trees.}
Suppose that a bipartite graph satisfies $M(t_j)\to J/N$.
Compactness of the unitary group gives a subsequence with $U(t_j)\to Q$.
Then $|Q_{uv}|^2=1/N$, $AQ=QA$, and
$\Sigma Q\Sigma=\overline Q$.
Thus $\widehat Q=\sqrt N\,Q$ has entries $\pm1$ within sublattices
and $\pm i$ across them \cite{GMR17}.
The $(a,b)$ entries in $A\widehat Q=\widehat QA$ equate sums
of $\deg a$ and $\deg b$ terms in the same set $\{\omega,-\omega\}$,
where $\omega\in\{1,i\}$.
Division by $\omega$ and reduction modulo $2$ give equal degree parity.

For a tree with at least two vertices, every degree is therefore odd.
If it has at least three vertices, counting leaves gives two leaves
$u,w$ with a common neighbor.
For $x=e_u-e_w$, one has $Ax=0$ and hence $Qx=x$.
Thus $Q_{uu}-Q_{uw}=1$.
The entry moduli imply $1\le2/\sqrt N$, so $N\le4$.
The only candidates are $K_1$, $P_2$ and $K_{1,3}$,
and each has instantaneous uniform mixing.

\textit{Products.}
The identity $M_{G\square H}(t)=M_G(t)\otimes M_H(t)$ gives
the stated criterion for products.
Summing it over $v'$ yields $\int M_G(t)_{uv}\,d\mu=1/|V(G)|$.
For $H=G$ and $(u',v')=(u,v)$, it also gives
\[
\int\left(M_G(t)_{uv}-\frac1{|V(G)|}\right)^2d\mu=0.
\]
Every entry is therefore uniform almost surely.
Since there are finitely many entries, $M_G(t)=J/|V(G)|$
holds simultaneously for $\mu$-almost every $t$.
Conversely, a point mass at an instantaneous uniform mixing time
works for $G\square G$.

\clearpage
\setcounter{section}{0}
\setcounter{equation}{0}
\setcounter{theorem}{0}
\setcounter{secnumdepth}{2}
\renewcommand{\theHequation}{S.\arabic{equation}}
\renewcommand{\thesection}{S\arabic{section}}
\renewcommand{\thesubsection}{\thesection.\arabic{subsection}}
\renewcommand{\theequation}{S\arabic{equation}}

\begin{center}
\textbf{\large Supplemental Material}
\end{center}

This Supplemental Material gives complete proofs of the results stated in the Letter.
We keep the notation of the Letter.
Equations and statements carry the prefix S(upplemental),
and Figs.~\ref{fig:collision} and \ref{fig:torus} refer to the Letter.

\section{Preliminaries}\label{S:sec:prelim}

\subsection{Quantum walks and averaged mixing}

Let $A$ be a real symmetric $n\times n$ matrix.
A graph adjacency matrix is one example.
Put
\begin{equation}
U(t)=e^{itA},\quad M(t)=U(t)\circ\overline{U(t)},
\end{equation}
where $\circ$ denotes the entrywise product.
The entry $M(t)_{uv}$ is the probability of finding at $u$
a walker started at $v$.
Unitarity makes $M(t)$ doubly stochastic.
For a Borel probability measure $\mu$ on $\R$, put
$\Mh[\mu]=\int_\R M(t)\,d\mu(t)$.
Throughout, $I$ and $J$ denote the identity and all-ones matrices.

\begin{definition}
The matrix $A$ admits \emph{uniform average mixing} if
$\Mh[\mu]=J/n$ for some Borel probability measure $\mu$.
It admits \emph{instantaneous uniform mixing} if $M(t_0)=J/n$
for some $t_0$.
It admits \emph{$\epsilon$-uniform mixing} if
\[
\inf_{t\in\R}\max_{u,v}\left|M(t)_{uv}-\frac1n\right|=0.
\]
Thus $\epsilon$-uniform mixing means approximation to arbitrary accuracy.
For a graph, we use its adjacency matrix.
\end{definition}

Since $A$ is real, $U(-t)=\overline{U(t)}$ and $M(-t)=M(t)$.
Pushing a time measure forward under $t\mapsto|t|$ therefore preserves
the averaged mixing matrix.
This also preserves finite support and absolute continuity.
Thus all existence statements can be realized with nonnegative readout times.

Let $A=\sum_r\theta_rE_r$ be the spectral decomposition with distinct
eigenvalues $\theta_r$.
The idempotents $E_r$ are real, and $U(t)=\sum_re^{it\theta_r}E_r$.
Pairing the terms indexed by $(r,s)$ and $(s,r)$ gives
\cite[Proposition~4]{BCM24}
\begin{equation}\label{S:eq:BCM}
\Mh[\mu]=\sum_rE_r\circ E_r+2\sum_{r<s}y_{rs}\,E_r\circ E_s,
\end{equation}
where the \emph{coherences} are
\begin{equation}\label{S:eq:coh}
y_{rs}=\int_\R\cos\bigl((\theta_r-\theta_s)t\bigr)\,d\mu(t).
\end{equation}
Thus $\Mh[\mu]$ depends only on the averaged cosines of the
Bohr frequencies $\theta_r-\theta_s$.

The uniformly coupled $XY$ Hamiltonian
$\hat H=\frac12\sum_{u=1}^{N-1}
(\sigma^x_u\sigma^x_{u+1}+\sigma^y_u\sigma^y_{u+1})$
preserves excitation number.
Its restriction to the single-excitation subspace is the adjacency
matrix of $P_N$ \cite{Christandl04,GKSS12}.
A time distribution makes the single-excitation populations uniform
for every initial site exactly when it gives uniform average mixing
on $P_N$.

\subsection{Bipartite matrices and moment pairs}

\begin{lemma}\label{S:lem:bip}
Let $A$ be a real symmetric $n\times n$ matrix.
Let $\Sigma=\operatorname{diag}(\sigma_u)$ with $\sigma_u\in\{1,-1\}$,
and assume $\Sigma A\Sigma=-A$.
This holds for a bipartite graph with its sublattice signature.
Write $A=\sum_{r=1}^d\theta_rE_r$ with $\theta_1>\dots>\theta_d$.
Put $\cR_+=\{r:\theta_r\ge0\}$ and define
\[
w(u,v)_r=
\begin{cases}
2(E_r)_{uv},&\theta_r>0,\\
(E_r)_{uv},&\theta_r=0.
\end{cases}
\]
Then $\theta_{d+1-r}=-\theta_r$ and $E_{d+1-r}=\Sigma E_r\Sigma$.
For every $t$,
\[
M(t)_{uv}=
\begin{dcases}
\left(\sum_{r\in\cR_+}w(u,v)_r\cos(t\theta_r)\right)^2,
&\sigma_u\sigma_v=1,\\
\left(\sum_{r\in\cR_+}w(u,v)_r\sin(t\theta_r)\right)^2,
&\sigma_u\sigma_v=-1.
\end{dcases}
\]
Define the symmetric moment matrices by
\[
\begin{aligned}
C[\mu]_{rs}&=\int_\R\cos(t\theta_r)\cos(t\theta_s)\,d\mu(t),\\
S[\mu]_{rs}&=\int_\R\sin(t\theta_r)\sin(t\theta_s)\,d\mu(t).
\end{aligned}
\]
Then
\begin{equation}\label{S:eq:pairing}
\Mh[\mu]_{uv}=
\begin{cases}
w(u,v)^{\mathsf T}C[\mu]w(u,v),&\sigma_u\sigma_v=1,\\
w(u,v)^{\mathsf T}S[\mu]w(u,v),&\sigma_u\sigma_v=-1.
\end{cases}
\end{equation}
\end{lemma}

\begin{proof}
The map $\Sigma$ sends the $\theta_r$-eigenspace onto the
$(-\theta_r)$-eigenspace.
This gives the spectral symmetry and
$(E_{d+1-r})_{uv}=\sigma_u\sigma_v(E_r)_{uv}$.

In the expansion of $U(t)_{uv}$, a positive eigenvalue and its negative
partner contribute
\[
\bigl(e^{it\theta_r}+\sigma_u\sigma_v e^{-it\theta_r}\bigr)(E_r)_{uv}.
\]
This equals $2\cos(t\theta_r)(E_r)_{uv}$ when $\sigma_u\sigma_v=1$,
and $2i\sin(t\theta_r)(E_r)_{uv}$ otherwise.
For a zero eigenvalue, $\Sigma E_r\Sigma=E_r$.
Its entry $(E_r)_{uv}$ therefore vanishes when $\sigma_u\sigma_v=-1$.

Thus $U(t)_{uv}$ is real or purely imaginary according to the
sublattice parity.
Taking its squared modulus gives the formula for $M(t)_{uv}$.
Integration gives \eqref{S:eq:pairing}.
\end{proof}

A \emph{symmetric pair} $(C,S)$ consists of two real symmetric matrices
indexed by $\cR_+$.
Define $\bM(C,S)$ by the right-hand side of \eqref{S:eq:pairing}.
The map $\bM$ is linear.
Put
\begin{equation}\label{S:eq:F}
F(t)=\left(
(\cos(t\theta_r)\cos(t\theta_s))_{r,s},
(\sin(t\theta_r)\sin(t\theta_s))_{r,s}
\right).
\end{equation}
Then $(C[\mu],S[\mu])=\int F\,d\mu$ and
$\Mh[\mu]=\bM(C[\mu],S[\mu])$.
In particular, $\bM(F(t))=M(t)$.
Writing $m=|\cR_+|$, the space of symmetric pairs has dimension $m(m+1)$.
The $m$ independent equations $C_{rr}+S_{rr}=1$ define an affine
space of dimension $m^2$ containing $F(\R)$.

\begin{lemma}\label{S:lem:bary}
Let $\mathcal X$ be a metric space and
$\Psi:\mathcal X\to\R^k$ a bounded continuous map.
For every Borel probability measure $\nu$ on $\mathcal X$,
\[
\int\Psi\,d\nu\in\conv\Psi(\mathcal X).
\]
If $\nu$ has full support, this barycenter belongs to
$\relint\conv\Psi(\mathcal X)$.
\end{lemma}

\begin{proof}
Put $x=\int\Psi\,d\nu$ and
$\mathcal C=\operatorname{cl}\conv\Psi(\mathcal X)$.
If $x\notin\mathcal C$, strict separation gives a vector $a$ with
\[
\langle a,x\rangle>\sup_{z\in\mathcal X}\langle a,\Psi(z)\rangle.
\]
Integration contradicts this inequality.
Hence $x\in\mathcal C$.

The relative interiors of a convex set and its closure agree
\cite[Theorem~6.3]{Rock70}.
Both claims follow immediately if $x\in\relint\mathcal C$.
Otherwise, a supporting hyperplane in the affine hull of $\mathcal C$
gives a linear functional $\langle a,\cdot\rangle$ that is nonconstant
on $\mathcal C$ and satisfies
\[
\langle a,y\rangle\le\langle a,x\rangle
\quad(y\in\mathcal C)
\]
\cite[Theorem~11.6]{Rock70}.
The continuous function $h(z)=\langle a,\Psi(z)-x\rangle$
is nonpositive and has integral zero.
Thus $\nu(\mathcal X')=1$ for the closed set $\mathcal X'=\{h=0\}$.

If $\nu$ has full support, continuity forces $h=0$ everywhere.
This contradicts nonconstancy of the functional on $\mathcal C$.
It proves the full-support assertion.

For the first assertion, induct on the affine dimension of
$\Psi(\mathcal X)$.
Dimension zero is immediate.
The image $\Psi(\mathcal X')$ lies in the supporting hyperplane,
so it has smaller affine dimension.
Apply the induction hypothesis to $\Psi|_{\mathcal X'}$
and the restriction of $\nu$.
It gives
\[
x\in\conv\Psi(\mathcal X')\subseteq\conv\Psi(\mathcal X).
\]
\end{proof}

\begin{theorem}\label{S:thm:reduction}
Let $A$ be as in Lemma~\ref{S:lem:bip}.
\begin{enumerate}[label=(\arabic*),leftmargin=2em,itemsep=1pt,topsep=2pt]
\item
The matrix $A$ admits uniform average mixing if and only if
$\bM(C,S)=J/n$ for some $(C,S)\in\conv F(\R)$.
\item
Let $\mathcal A$ be an affine space containing $F(\R)$,
and suppose that $\bM$ is injective on $\mathcal A$.
Assume that $(C^\star,S^\star)\in\mathcal A$ satisfies
$\bM(C^\star,S^\star)=J/n$.
Then every uniformly mixing measure has this moment pair.
Moreover, uniform average mixing occurs if and only if
$(C^\star,S^\star)\in\conv F(\R)$.
\end{enumerate}
\end{theorem}

\begin{proof}
(1)
For a uniformly mixing measure, Lemma~\ref{S:lem:bary} puts
$\int F\,d\mu$ in $\conv F(\R)$.
Its image under $\bM$ is $J/n$.
Conversely, a finite convex combination $\sum_j\mu_jF(t_j)$
with image $J/n$ gives the measure $\sum_j\mu_j\delta_{t_j}$.

(2)
The moment pair of a uniformly mixing measure belongs to
$\conv F(\R)\subseteq\mathcal A$.
It has the same image as the target, so injectivity makes
the two pairs equal.
The final equivalence follows from (1).
\end{proof}

Theorem~\ref{S:thm:reduction} reformulates \eqref{S:eq:BCM}
as a finite-dimensional convex problem.
Injectivity of $\bM$ makes the moments of every uniformly mixing
measure unique.
We now establish this property for paths.

\section{Forced coherences on paths}\label{S:sec:forced}

\subsection{Spectrum and target}

Let $q\ge3$ and $N=q-1$.
Let $A$ be the adjacency matrix of $P_N$ on vertices $1,\dots,N$.
Its eigenvalues and normalized eigenvectors are
\begin{equation}\label{S:eq:spec}
\theta_r=2\cos\frac{r\pi}{q},
\quad (v_r)_k=\sqrt{\frac2q}\sin\frac{rk\pi}{q},
\end{equation}
for $1\le r,k\le N$.
The spectral idempotents are $E_r=v_rv_r^{\mathsf T}$.
Let $\Iop_{k\ell}=\delta_{k+\ell,q}$ be the reversal matrix.
The identities
\[
\theta_{q-r}=-\theta_r,\quad
(v_{q-r})_k=(-1)^{k+1}(v_r)_k
\]
give Lemma~\ref{S:lem:bip} with $d=N$ and $\sigma_k=(-1)^{k+1}$.
The nonnegative and positive mode indices are respectively
\[
\cR=\{1,\dots,\lfloor q/2\rfloor\},
\quad
\cR^\times=\{1,\dots,\lceil q/2\rceil-1\}.
\]

\begin{lemma}\label{S:lem:alt}
$\sum_{r=1}^N(-1)^rE_r=-\Iop$, and consequently
\[
\sum_{r\equiv s\ (\mathrm{mod}\ 2)}E_r\circ E_s=\frac12(I+\Iop).
\]
\end{lemma}

\begin{proof}
Reflection of the sine eigenvector gives
\[
(\Iop v_r)_k=(v_r)_{q-k}=(-1)^{r-1}(v_r)_k.
\]
Since the eigenvectors form an orthonormal basis,
$\Iop=\sum_r(-1)^{r-1}E_r$.
For the second identity, use
\[
\sum_{r,s}E_r\circ E_s=I,\quad
\sum_{r,s}(-1)^{r+s}E_r\circ E_s=\Iop.
\]
Their half-sum selects pairs of equal parity.
\end{proof}

\begin{proposition}[{\cite[Lemma~4.3]{God13}}]\label{S:prop:D0}
\[
D_0:=\sum_rE_r\circ E_r=\frac{2J+I+\Iop}{2q}.
\]
\end{proposition}

\begin{proof}
The sine formula for the eigenvectors gives
\[
(D_0)_{k\ell}
=\frac1{q^2}\sum_{r=0}^{q-1}
\left(1-\cos\frac{2rk\pi}{q}\right)
\left(1-\cos\frac{2r\ell\pi}{q}\right).
\]
Use $\sum_{r=0}^{q-1}\cos(2rj\pi/q)=q[q\mid j]$
and the product-to-sum identity.
The two single-cosine sums vanish because $1\le k,\ell<q$.
The remaining terms give
\[
(D_0)_{k\ell}
=\frac1q+\frac1{2q}[q\mid k-\ell]
+\frac1{2q}[q\mid k+\ell].
\]
In this index range, the last two conditions are $k=\ell$
and $k+\ell=q$, respectively.
\end{proof}

\begin{corollary}\label{S:cor:parity}
\[
D_0-\frac2N\sum_{r<s,\ s-r\text{ even}}E_r\circ E_s=\frac1NJ.
\]
\end{corollary}

\begin{proof}
Lemma~\ref{S:lem:alt} gives
\[
\sum_{r<s,\ s-r\text{ even}}E_r\circ E_s
=\frac12\left(\frac12(I+\Iop)-D_0\right).
\]
The asserted left-hand side therefore equals
$\frac qN D_0-\frac1{2N}(I+\Iop)$.
Proposition~\ref{S:prop:D0} reduces it to $J/N$.
\end{proof}

Call a real-valued function $y$ on $\{1,\dots,N\}^2$
\emph{admissible} if
\[
\begin{gathered}
y(a,a)=1,\quad y(a,b)=y(b,a)=y(q-a,q-b),\\
y(q-a,b)=y(a,q-b).
\end{gathered}
\]
For admissible $y$, define the symmetric pair
\begin{equation}\label{S:eq:pairy}
\begin{aligned}
C^y_{rs}&=\frac{y(r,s)+y(r,q-s)}2,\\
S^y_{rs}&=\frac{y(r,s)-y(r,q-s)}2
\end{aligned}
\end{equation}
for $r,s\in\cR$.
The function $y_t(a,b)=\cos((\theta_a-\theta_b)t)$ is admissible.
Since $\theta_{q-s}=-\theta_s$, the product-to-sum identities give
$(C^{y_t},S^{y_t})=F(t)$.

Conversely, an admissible function is determined by its pair.
Indeed,
\[
y(r,s)=C^y_{rs}+S^y_{rs},\quad
y(r,q-s)=C^y_{rs}-S^y_{rs}
\]
for $r,s\in\cR$.
Admissibility recovers the remaining entries.

\begin{lemma}\label{S:lem:expansion}
For admissible $y$,
\[
\bM(C^y,S^y)=\sum_{a,b=1}^Ny(a,b)\,E_a\circ E_b.
\]
\end{lemma}

\begin{proof}
Fix $u,v$ and put $\sigma=\sigma_u\sigma_v$ and $\beta_a=(E_a)_{uv}$.
Then $\beta_{q-a}=\sigma\beta_a$.
Partition the mode indices into the two-element sets $\{r,q-r\}$
with $r<q/2$, together with the zero-mode singleton when $q$ is even.

For $r,s<q/2$, admissibility gives
\[
\begin{aligned}
&\sum_{a\in\{r,q-r\}}\sum_{b\in\{s,q-s\}}y(a,b)\beta_a\beta_b\\
&\quad=2\bigl(y(r,s)+\sigma y(r,q-s)\bigr)\beta_r\beta_s.
\end{aligned}
\]
Since $w(u,v)_r=2\beta_r$, this is the $(r,s)$ summand
of $w^{\mathsf T}C^y w$ for $\sigma=1$ or
$w^{\mathsf T}S^y w$ for $\sigma=-1$.

For even $q$, write $z=q/2$.
A block with $r=z>s$ contributes
$(1+\sigma)y(z,s)\beta_z\beta_s$.
The transposed block has the same formula.
The singleton block contributes $\beta_z^2$.
Here
\[
C^y_{zs}=y(z,s),\quad S^y_{zs}=0,\quad
C^y_{zz}=1,\quad S^y_{zz}=0.
\]
Also, $w(u,v)_z=\beta_z$, which vanishes when $\sigma=-1$.
Thus these blocks agree with the same quadratic-form expression.
Summing all blocks proves the identity.
\end{proof}

\begin{proposition}\label{S:prop:target}
Define
\[
\varkappa(j)=
\begin{cases}
1,&j=0,\\
-1/N,&1\le|j|\le N-1\text{ and }j\text{ even},\\
0,&1\le|j|\le N-1\text{ and }j\text{ odd}.
\end{cases}
\]
Then $y^\star(a,b)=\varkappa(b-a)$ is admissible.
Its pair $(C^\star,S^\star)$ satisfies
$\bM(C^\star,S^\star)=J/N$.
\end{proposition}

\begin{proof}
The function $\varkappa$ is even and $\varkappa(0)=1$.
Reversing both indices changes $b-a$ to $a-b$.
For the cross identity, the arguments are $a+b-q$ and $q-a-b$.
They are negatives of one another.
All arguments lie in $[-(N-1),N-1]$, so $y^\star$ is admissible.
Lemma~\ref{S:lem:expansion} gives
\[
\bM(C^\star,S^\star)
=D_0+2\sum_{a<b}\varkappa(b-a)E_a\circ E_b.
\]
Corollary~\ref{S:cor:parity} identifies this matrix as $J/N$.
\end{proof}

\subsection{Uniqueness of the coherences}

For $0\le j\le2q$, define $c_j\in\R^N$ by
$(c_j)_u=\cos(ju\pi/q)$.
Then $c_{2q-j}=c_j$.
For $1\le k<l\le N$, put $\psi(k,l)=(q-l,q-k)$.

\begin{lemma}\label{S:lem:rankone}
For $1\le k<l\le N$,
\[
E_k\circ E_l=\frac1{q^2}\,(c_{l-k}-c_{k+l})(c_{l-k}-c_{k+l})^{\mathsf T}.
\]
Consequently $E_{q-l}\circ E_{q-k}=E_k\circ E_l$, and $\psi$ preserves both $l-k$ and $\theta_k-\theta_l$.
\end{lemma}

\begin{proof}
Since $E_r=v_rv_r^{\mathsf T}$,
\[
E_k\circ E_l=(v_k\circ v_l)(v_k\circ v_l)^{\mathsf T}.
\]
The product-to-sum identity gives
\[
(v_k)_u(v_l)_u
=\frac1q\left(\cos\frac{(l-k)u\pi}{q}
-\cos\frac{(k+l)u\pi}{q}\right).
\]
This proves the formula.
The reflected pair has the same difference and sum $2q-(k+l)$.
Since $c_{2q-j}=c_j$, its rank-one matrix agrees.
Finally, $\theta_{q-l}-\theta_{q-k}=-\theta_l+\theta_k$,
so $\psi$ preserves the Bohr frequency.
\end{proof}

\begin{lemma}\label{S:lem:rank}
Choose one pair from each $\psi$-orbit in $\{(k,l):1\le k<l\le N\}$. The corresponding matrices
$E_k\circ E_l$ are linearly independent.
\end{lemma}

\begin{proof}
For an integer $j$, geometric sums give
\[
\sum_{u=1}^{q-1}\cos\frac{ju\pi}{q}
=\begin{cases}
q-1,&2q\mid j,\\
-1,&j\text{ even and }2q\nmid j,\\
0,&j\text{ odd}.
\end{cases}
\]
For odd $j$, the zero sum also follows by pairing $u$ with $q-u$.
The product-to-sum identity gives
\begin{equation}\label{S:eq:gram}
\langle c_a,c_j\rangle
=\frac q2[a=j]-[a\equiv j\pmod2]
\end{equation}
for $0\le a\le q$ and $1\le j\le N$.
In this range, $a+j$ is never divisible by $2q$.
The difference $a-j$ is divisible by $2q$ only when $a=j$.

Choose each $\psi$-orbit representative with $c=k+l\le q$,
and put $b=l-k$.
Then $1\le b<c\le q$ and $b\equiv c\pmod2$.
The formulas $k=(c-b)/2$ and $l=(c+b)/2$ show that distinct orbits
give distinct pairs $(b,c)$.
Put $w_{bc}=c_b-c_c$.
The parity terms in \eqref{S:eq:gram} cancel, giving
\[
\langle w_{bc},c_j\rangle=\frac q2([b=j]-[c=j]).
\]

Suppose that $\sum_{(b,c)}z_{bc}w_{bc}w_{bc}^{\mathsf T}=0$.
Testing on $c_j,c_{j'}$ and cancelling $q^2/4$ gives
\[
\sum_{(b,c)}z_{bc}([b=j]-[c=j])([b=j']-[c=j'])=0
\]
for $1\le j,j'\le N$.
Let $L=\sum_{(b,c)}z_{bc}(e_b-e_c)(e_b-e_c)^{\mathsf T}$ on $\R^q$.
This is a graph Laplacian with possibly signed real edge weights.
Its principal submatrix on $\{1,\dots,N\}$ vanishes.
For $c\le N$, the $(b,c)$ entry is $-z_{bc}$, so these weights vanish.
Only edges $\{b,q\}$ can remain.
Their weights are the remaining diagonal entries $L_{bb}=z_{bq}$,
so they also vanish.
Lemma~\ref{S:lem:rankone} proves the claimed independence.
\end{proof}

\begin{proposition}\label{S:prop:injective}
Let $\mathcal A$ be the affine span of $F(\R)\cup\{(C^\star,S^\star)\}$. Then $\bM$ is injective on $\mathcal A$.
\end{proposition}

\begin{proof}
Admissible functions form an affine space, and
$y\mapsto(C^y,S^y)$ is affine.
Its image contains $F(\R)$ and the target.
Thus every point of $\mathcal A$ comes from an admissible function.

Suppose that two such pairs have the same image under $\bM$.
Their diagonal function values both equal $1$.
Lemma~\ref{S:lem:expansion} gives
\[
\sum_{k<l}(y-y')(k,l)E_k\circ E_l=0.
\]
Admissibility makes the coefficients constant on $\psi$-orbits.
Grouping by orbit multiplies each coefficient by the orbit size,
which is $1$ or $2$.
Lemmas~\ref{S:lem:rankone} and~\ref{S:lem:rank} force all coefficients
to vanish.
Hence $y=y'$, and the moment pairs agree.
\end{proof}

\begin{theorem}[Forced coherences]\label{S:thm:forced}
If $\Mh[\mu]=\frac1NJ$ for $P_N$, then $(C[\mu],S[\mu])=(C^\star,S^\star)$. Equivalently, for $1\le k<l\le N$,
\begin{equation}\label{S:eq:forced}
\int_\R\cos\bigl((\theta_k-\theta_l)t\bigr)\,d\mu(t)=
\begin{cases}\displaystyle-\frac1N, & l-k\text{ even},\\[6pt] 0, & l-k\text{ odd}.\end{cases}
\end{equation}
\end{theorem}

\begin{proof}
Apply Theorem~\ref{S:thm:reduction}(2) to the affine space in
Proposition~\ref{S:prop:injective} and the target in
Proposition~\ref{S:prop:target}.
This gives equality of the moment pairs.
For $y(a,b)=\int_\R\cos((\theta_a-\theta_b)t)\,d\mu(t)$,
the pair is $(C[\mu],S[\mu])$.
An admissible function is determined by its pair.
Thus the moment equality is equivalent to $y=y^\star$.
\end{proof}

\begin{corollary}\label{S:cor:reduction}
$P_{q-1}$ admits uniform average mixing if and only if $(C^\star,S^\star)\in\conv F(\R)$.
\end{corollary}

\begin{proof}
Apply Theorem~\ref{S:thm:reduction}(2) and
Propositions~\ref{S:prop:target} and~\ref{S:prop:injective}.
\end{proof}

The reflection parity of $v_r$ is $(-1)^{r-1}$.
Thus \eqref{S:eq:forced} requires zero cosine coherence between modes
of opposite mirror parity.
Between distinct modes of equal parity, it requires $-1/N$
independently of the Bohr frequency.

For $q=4$, the gaps $\sqrt2$ and $2\sqrt2$ have odd and even
index differences, respectively.
Hence
\[
\int_\R\cos(\sqrt2t)\,d\mu(t)=0,\quad
\int_\R\cos(2\sqrt2t)\,d\mu(t)=-\frac13.
\]
These are the values in \cite[Corollary~9]{BCM24}.
The case $q=5$ agrees with \cite[\S4.3]{BCM24}.
Theorem~\ref{S:thm:forced} establishes uniqueness for every path.

\section{Frequency collisions}\label{S:sec:collisions}

\begin{definition}\label{S:def:collision}
A \emph{parity collision} at the modulus $q$ consists of pairs $1\le k<l\le N$ and $1\le k'<l'\le N$ with
$\theta_k-\theta_l=\theta_{k'}-\theta_{l'}$ and $l-k\not\equiv l'-k'\pmod2$.
\end{definition}

\begin{corollary}\label{S:cor:collision}
If there is a parity collision at $q$, then $P_{q-1}$ admits no uniform average mixing.
\end{corollary}

\begin{proof}
By Theorem~\ref{S:thm:forced} the moment at the common Bohr frequency would be both $-\frac1N$ and $0$.
\end{proof}

A \emph{relation} is an identity $\sum_i a_i\eta_i=0$
with positive integers $a_i$ and distinct roots of unity $\eta_i$.
Its \emph{weight} is $\sum_i a_i$.
It is \emph{minimal} if no nonempty proper submultiset of its terms
has sum zero.
Every relation decomposes into minimal relations.
A \emph{rotation} multiplies all terms by the same root of unity.

For a prime $p$, let $R_p$ denote the relation
\[
\sum_{i=0}^{p-1}\zeta_p^{\,i}=0,
\quad \zeta_m=e^{2\pi i/m}.
\]
For the mixed types below, take $p\in\{5,7\}$.
Choose $j$ rotated copies of $R_3$, each sharing one root with $R_p$.
Require different copies to share different roots of $R_p$.
Subtract these relations from $R_p$ and absorb the minus signs
into the remaining roots.
Denote any resulting relation by $(R_p{:}jR_3)$.
Equivalently, each selected term $\eta$ is replaced by
$-\eta\zeta_3$ and $-\eta\zeta_3^2$.

\begin{theorem}[{\cite[Theorem~3 and Table~1]{PR98}}]
\label{S:thm:PR}
Up to rotation, every minimal relation of weight at most $8$
has one of the following types:
\begin{align*}
& R_2,\quad R_3,\quad R_5,\quad (R_5{:}R_3),\quad \\
& R_7,\quad (R_5{:}2R_3),\quad (R_5{:}3R_3),\quad (R_7{:}R_3).
\end{align*}
Their respective weights are $2,3,5,6,7,7,8,8$.
\end{theorem}

\begin{proof}
We give the weight-at-most-eight argument to make the classification
used below explicit.
Rotate the relation so that one term is $1$.
Let $m$ be the least common multiple of the orders of its terms.
The positivity of the coefficients implies $m>1$.

First, $m$ is squarefree.
Suppose that $p^2\mid m$ for a prime $p$.
The extension
$\Q(\zeta_m)/\Q(\zeta_{m/p})$ has degree $p$.
Grouping powers of $\zeta_m$ by their exponents modulo $p$ writes
the relation as
\[
\sum_{j=0}^{p-1}\zeta_m^j A_j=0,
\quad A_j\in\Q(\zeta_{m/p}).
\]
Linear independence forces every $A_j$ to vanish.
Each nonempty group is therefore a subrelation.
Minimality permits only one nonempty group.
The term $1$ puts this group at $j=0$, so all terms have orders
dividing $m/p$.
This contradicts the definition of $m$.

Now let $p\mid m$ and put $n=m/p$.
Since $m$ is squarefree, every term has a unique form
$\zeta_p^j\rho$ with $\rho^n=1$.
The extension $\Q(\zeta_m)/\Q(\zeta_n)$ has degree $p-1$.
Thus the minimal polynomial of $\zeta_p$ over $\Q(\zeta_n)$ is
$1+X+\dots+X^{p-1}$.
Writing the relation as $\sum_{j=0}^{p-1}\zeta_p^j A_j=0$
shows that all group sums $A_j$ are equal.
No group can be empty.
Otherwise all group sums would be zero, and minimality would again
force all terms into the group containing $1$, contradicting $p\mid m$.
Consequently every prime dividing $m$ is at most the weight.
For weight at most $8$, it follows that $m\mid210$.

We also need the minimal relations among sixth roots of unity.
An antipodal pair gives a relation of type $R_2$.
If no antipodal pair occurs, write the sum as
$c_0+c_1\zeta_3+c_2\zeta_3^2$ with signed integers $c_j$.
For each $j$, only one of $\zeta_3^j$ and $-\zeta_3^j$ can occur.
Vanishing gives $c_0=c_1=c_2$.
Minimality then forces the common coefficient to be $1$ or $-1$.
Thus the only minimal types among sixth roots are $R_2$ and $R_3$.

Suppose that the largest prime dividing $m$ is $7$.
All seven groups are nonempty, and their total weight is at most $8$.
There is a singleton group.
Its sum is a root of unity.
Dividing all group sums by this root makes their common value $1$.
At weight $7$, every group is a singleton and the type is $R_7$.
At weight $8$, exactly one group contains two roots.
If two unit complex numbers sum to $1$, they are
$-\zeta_3$ and $-\zeta_3^2$.
Hence this case has type $(R_7{:}R_3)$.

Suppose instead that the largest prime is $5$.
The five groups have total weight at most $8$, so again one is
a singleton.
Normalize their common sum to $1$.
Each group now consists of sixth roots, because $m/5\mid6$.
It has no nonempty zero-sum submultiset, by minimality of the
original relation.
Appending $-1$ therefore produces a minimal relation:
a proper subrelation omitting $-1$ would already be a zero-sum
submultiset, and one containing $-1$ would leave such a submultiset.
The sixth-root classification shows that each group is either
the singleton $1$ or the pair $-\zeta_3,-\zeta_3^2$.
At most three groups are pairs.
This gives $R_5$ and $(R_5{:}jR_3)$ for $1\le j\le3$.

If the largest prime is at most $3$, then $m\mid6$.
The remaining possibilities are $R_2$ and $R_3$.
\end{proof}

Conway and Jones \cite{CJ76} established the classification
for weight at most $9$, building on Mann \cite{Mann65}.

\begin{theorem}[Collision moduli]\label{S:thm:collset}
Let $q\ge3$. There is a parity collision at $q$ if and only if
at least one of the following conditions holds:
\[
6\mid q \text{ and } q\ge12,\quad
15\mid q,\quad
21\mid q.
\]
For every such $q$, the path $P_{q-1}$ admits no uniform average mixing.
\end{theorem}

\begin{proof}
\emph{Sufficiency.}
Suppose that $15\mid q$ or $21\mid q$ with an even cofactor.
Then $6\mid q$ and $q\ge30$.
Thus, after treating the multiples of $6$, we may assume
that the cofactors in the other two families are odd.

First suppose that $6\mid q$ and $q\ge12$.
Put $h=q/3$. Then $h$ is even and $h\ge4$.
For $1\le k\le h/2$, we have
\[
\begin{aligned}
\theta_k-\theta_{k+h}
&=4\sin\frac{(2k+h)\pi}{2q}\sin\frac{\pi}{6}\\
&=\theta_{h-k}\\
&=\theta_{h-k}-\theta_{q/2}.
\end{aligned}
\]
Here we used $h\pi/(2q)=\pi/6$, $q-h=2h$ and $\theta_{q/2}=0$.
The pair $(k,k+h)$ has even index difference $h$.
The pair $(h-k,q/2)$ has index difference $h/2+k$.
Choose $k\in\{1,2\}$ so that $h/2+k$ is odd.
This gives a parity collision.
Figure~\ref{fig:collision} shows the case $q=12$ and $k=1$.

Next suppose that $q=15a$ with $a$ odd.
Put $\vartheta=a\pi/q=\pi/15$.
We have
\[
\cos\vartheta-\cos4\vartheta
=2\sin\frac{5\vartheta}{2}\sin\frac{3\vartheta}{2}
=\sin\frac{\pi}{10}
=\cos\frac{2\pi}{5}.
\]
Also, $3\vartheta=\pi/5$ and $5\vartheta=\pi/3$.
Therefore
\begin{align*}
&(\theta_a-\theta_{3a})-(\theta_{4a}-\theta_{5a}) \\
&=2(\cos\vartheta-\cos4\vartheta)
  -2\cos3\vartheta+2\cos5\vartheta\\
&=2\cos\frac{2\pi}{5}-2\cos\frac{\pi}{5}+1\\
&=0.
\end{align*}
The last equality follows from
\[
1+2\cos\frac{2\pi}{5}+2\cos\frac{4\pi}{5}=0
\]
and $\cos(4\pi/5)=-\cos(\pi/5)$.
The index differences are $2a$ and $a$.
Since $a$ is odd, this is a parity collision.

Finally suppose that $q=21a$ with $a$ odd.
Then
\begin{align*}
&(\theta_{3a}-\theta_{6a})-(\theta_{7a}-\theta_{9a})\\
&=2\left(
  \cos\frac{\pi}{7}-\cos\frac{2\pi}{7}+\cos\frac{3\pi}{7}
  \right)-1\\
&=0.
\end{align*}
Indeed,
\[
1+2\sum_{j=1}^{3}\cos\frac{2j\pi}{7}=0,
\]
while $\cos(4\pi/7)=-\cos(3\pi/7)$
and $\cos(6\pi/7)=-\cos(\pi/7)$.
The index differences are $3a$ and $2a$.
Thus this is also a parity collision.

\emph{Necessity.}
Let $(k,l)$ and $(k',l')$ form a parity collision.
Then the four signed terms
\[
\theta_k,\quad -\theta_l,\quad -\theta_{k'},\quad \theta_{l'}
\]
sum to $0$.
Recall that $r\mapsto\theta_r$ is injective on $\{1,\dots,N\}$
and that $\theta_{q-j}=-\theta_j$.
We express each signed term in roots of unity and apply the
classification of minimal relations.

(i) \emph{No two of the four terms sum to $0$.}
Suppose otherwise. The remaining two terms also sum to $0$.
Since $\theta_k\ne\theta_l$ and $\theta_{k'}\ne\theta_{l'}$,
we must have
\[
(k',l')=(k,l)
\quad\text{or}\quad
(k',l')=\psi(k,l).
\]
In either case, $l'-k'=l-k$.
This contradicts the assumption of a parity collision.
In particular, at most one of the four indices equals $q/2$.

(ii) \emph{Construction of a relation.}
Each signed term $\pm\theta_j$ can be written as
\[
\zeta_{2q}^{x}+\zeta_{2q}^{-x},
\]
where $x=j$ for the positive sign and $x=j+q$ for the negative sign.
We regard $x$ as an element of $\mathbb Z/2q\mathbb Z$.
The set $\Lambda=\{x,-x\}$ has two elements and avoids $0$ and $q$.

Omit the term with index $q/2$ if it is present.
Collect the sets $\Lambda$ of the remaining terms into a multiset
$\mathcal E$.
Then $\mathcal E$ has weight $8$ or $6$ and satisfies
\[
\sum_{e\in\mathcal E}\zeta_{2q}^{e}=0.
\]

(iii) \emph{No two elements of $\mathcal E$ differ by $q$.}
Suppose that $e\in\Lambda$ and $e+q\in\Lambda'$ for two of the
sets used to construct $\mathcal E$.
Both $\Lambda$ and $\Lambda'+q$ are closed under negation
and contain $e$.
Since $e\ne-e$, both sets equal $\{e,-e\}$.
Hence $\Lambda'=\Lambda+q$.

If $\Lambda'=\Lambda$, then $2x\equiv q\pmod{2q}$.
The corresponding index is therefore $q/2$, which was excluded.
Otherwise, the sets correspond to two different terms.
The second term is the negative of the first, contradicting (i).

(iv) \emph{Divisibility.}
Decompose $\mathcal E$ into minimal relations.
By (iii), none is a rotation of $R_2$.

A rotation of $R_p$ contains two terms with ratio $\zeta_p$.
If its terms are $2q$th roots of unity, then $p\mid2q$.

Now consider a rotated relation of type $(R_p{:}jR_3)$,
where either $p=5$ and $1\le j\le3$, or $p=7$ and $j=1$.
It contains an unchanged term $\eta\zeta_p^{i'}$
and a replacement term $-\eta\zeta_3\zeta_p^i$
with $i\not\equiv i'\pmod p$.
The ratio is $-\zeta_3\zeta_p^{i-i'}$.
Its factors have orders $6$ and $p$.
Since $\gcd(6,p)=1$, their product has order $6p$.
Thus $30\mid2q$ when $p=5$, and $42\mid2q$ when $p=7$.

The multiset $\mathcal E$ has weight $8$ or $6$,
and no minimal component has weight $2$.
By Theorem~\ref{S:thm:PR}, the possible component weights are
\[
8,\quad 5+3,\quad 6,\quad 3+3.
\]

A component of weight $8$ has type $(R_5{:}3R_3)$
or $(R_7{:}R_3)$.
Hence $15\mid q$ or $21\mid q$.

For component weights $5+3$, the components have types $R_5$ and $R_3$.
Thus $5\mid2q$ and $3\mid2q$, giving $15\mid q$.
A component of weight $6$ has type $(R_5{:}R_3)$
and also gives $15\mid q$.

For component weights $3+3$, both components have type $R_3$.
Thus $3\mid q$.
In this case $\mathcal E$ has weight $6$.
By its construction in (ii), one of the four indices equals $q/2$.
Consequently $q$ is even, so $6\mid q$.

It remains to exclude $q=6$.
In this case
\[
(\theta_1,\dots,\theta_5)
=(\sqrt3,1,0,-1,-\sqrt3).
\]
The gaps with odd index difference are
\[
\sqrt3-1,\quad 1,\quad \sqrt3+1.
\]
Those with even index difference are
\[
\sqrt3,\quad 2,\quad 2\sqrt3.
\]
These two sets are disjoint, so there is no parity collision.
\end{proof}

\section{Construction by a phase ramp}\label{S:sec:construction}

\subsection{Torus filling}

For $m\ge1$, let $\zeta_m=e^{2\pi i/m}$. Then
\[
\theta_r=\zeta_{2q}^r+\zeta_{2q}^{-r}.
\]
Let $\mathbb L=\Q(\theta_1)$.
This is the maximal real subfield of $\Q(\zeta_{2q})$.
Its degree is $D=\frac12\varphi(2q)$, where $\varphi$ is Euler's
totient function.

\begin{lemma}\label{S:lem:cheb}
The numbers $1,\theta_1,\dots,\theta_{D-1}$ form a $\Q$-basis
of $\mathbb L$.
\end{lemma}

\begin{proof}
The identities
\[
\theta_1\theta_r=\theta_{r+1}+\theta_{r-1},
\quad
\theta_0=2
\]
show that each $\theta_r$ with $r\ge1$ is a monic integer polynomial
of degree $r$ in $\theta_1$.
Thus $1,\theta_1,\dots,\theta_{D-1}$ have the same span as
$1,\theta_1,\dots,\theta_1^{D-1}$.
The latter form a $\Q$-basis of $\mathbb L$ because $\theta_1$
has degree $D$ over $\Q$.
\end{proof}

\begin{theorem}\label{S:thm:torus}
Let $q\ge3$.
The numbers $\theta_r$ with $r\in\cR^\times$ are linearly independent
over $\Q$ if and only if $q$ is a power of two, a prime,
or twice a prime.
\end{theorem}

\begin{proof}
\emph{Necessity.}
Independence forces $|\cR^\times|\le D$.

For odd $q$, this inequality reads $\varphi(q)\ge q-1$.
Hence $q$ is prime.

For $q=2c$ with $c\ge3$ odd, it reads $\varphi(c)\ge c-1$.
Hence $c$ is prime.

For $q=2^ac$ with $a\ge2$ and $c$ odd, it reads
\[
2^{a-1}\bigl(c-\varphi(c)\bigr)\le1.
\]
This forces $c=1$.

\emph{Sufficiency.}
If $q=2^a$, then $D=|\cR^\times|+1$.
The result follows from Lemma~\ref{S:lem:cheb}.

It remains to consider $q=p$ and $q=2p$ with $p$ an odd prime.
In these cases $D=|\cR^\times|$.
Let $\mathcal V$ be the $\Q$-span of the numbers $\theta_r$
with $r\in\cR^\times$.
It suffices to show that $1\in\mathcal V$.
Indeed, Lemma~\ref{S:lem:cheb} then gives $\mathcal V=\mathbb L$.
The $D$ spanning numbers are therefore independent.

Since $\zeta_p=\zeta_{2q}^{2q/p}$, pairing $j$ with $p-j$ in
$\sum_{j<p}\zeta_p^{\,j}=0$ gives
\[
1=-\sum_{j=1}^{(p-1)/2}\theta_{2qj/p}.
\]
Each index $2qj/p$ lies strictly between $0$ and $q$
and differs from $q/2$.
The identity $\theta_{q-r}=-\theta_r$ therefore rewrites each
summand as $\pm\theta_{r'}$ with $r'\in\cR^\times$.
Thus $1\in\mathcal V$.
\end{proof}

\begin{lemma}\label{S:lem:kronecker}
Let $\T=\R/2\pi\Z$, and let $\cT\subseteq\T^{\cR}$ be the closure of
\[
\{(t\theta_r)_{r\in\cR}:t\in\R\}.
\]
Suppose that $q$ is a power of two, a prime, or twice a prime.
Then
\[
\cT=\{\vartheta\in\T^{\cR}:\vartheta_{q/2}=0\},
\]
where the coordinate condition is omitted for odd $q$.
In particular, $\cT$ contains the \emph{phase ramp}
\[
\gamma(\xi)=(\alpha_r\xi)_{r\in\cR},
\quad
\alpha_r=\frac q2-r,
\quad
\xi\in\R/4\pi\Z.
\]
\end{lemma}

\begin{proof}
For even $q$, the coordinate $t\theta_{q/2}$ vanishes.
The remaining coordinates are dense in $\T^{\cR^\times}$
by Theorem~\ref{S:thm:torus} and Kronecker's theorem
\cite[Chap.~XXIII]{HW08}.
For odd $q$, this gives density in all of $\T^{\cR}$.
The phase ramp belongs to $\cT$ because $\alpha_{q/2}=0$
when $q$ is even.
\end{proof}

\subsection{The slope density}

We take densities on $\R/4\pi\Z$ with respect to normalized Haar
measure.
For integers $j$, write
\[
\widehat f(j)=\int f(\xi)e^{-ij\xi}\,d\xi,
\]
where $d\xi$ denotes this normalized measure.
Only integer frequencies are needed below.

\begin{proposition}\label{S:prop:measure}
Let $q\ge4$ and $g=\lfloor\frac{N-1}{2}\rfloor\ge1$.
For $\lambda>0$, put
\begin{equation}\label{S:eq:density}
f_\lambda(\xi)
=1-\frac{2\lambda}{N}\sum_{j=1}^g\cos(2j\xi).
\end{equation}
Then
\[
\widehat{f_\lambda}(0)=1,
\quad
\widehat{f_\lambda}(j)=\lambda\varkappa(j)
\quad (1\le|j|\le N-1).
\]
Moreover,
\[
\min f_\lambda=1-\frac{2g\lambda}{N}.
\]
Thus $f_\lambda$ is strictly positive exactly when
$\lambda<\frac{N}{2g}$.
This upper bound exceeds $1$.
\end{proposition}

\begin{proof}
The frequencies $2j$ with $1\le j\le g$ are precisely the even
integers in $[1,N-1]$.
Orthogonality gives the stated Fourier coefficients.
The Dirichlet kernel
\[
1+2\sum_{j=1}^g\cos(2j\xi)
\]
has maximum $2g+1$ at $\xi=0$.
This gives the formula for $\min f_\lambda$.
Finally, $2g+1\in\{N-1,N\}$, so $2g<N$.
\end{proof}

\begin{proposition}\label{S:prop:phases}
Let $q\ge4$ be as in Lemma~\ref{S:lem:kronecker},
and let $0<\lambda<\frac{N}{2g}$.
Put
\[
\Pi_{rs}=[r=s=q/2],
\quad
(C^\circ,S^\circ)
=\left(\frac12I+\frac12\Pi,\frac12I-\frac12\Pi\right).
\]
This is the moment pair of Haar measure on $\cT$.

Suppose that $\xi$ has density $f_\lambda$, and put
$\Phi=\gamma(\xi)$.
Then $\Phi$ takes values in $\cT$ and satisfies
\[
\bigl(C[\Phi],S[\Phi]\bigr)
=(1-\lambda)(C^\circ,S^\circ)
+\lambda(C^\star,S^\star),
\]
where
\[
C[\Phi]_{rs}=\mathbb E[\cos\Phi_r\cos\Phi_s],
\quad
S[\Phi]_{rs}=\mathbb E[\sin\Phi_r\sin\Phi_s].
\]
\end{proposition}

\begin{proof}
Lemma~\ref{S:lem:kronecker} gives $\Phi\in\cT$.
Also,
\[
\Phi_r\pm\Phi_s=(\alpha_r\pm\alpha_s)\xi,
\]
The coefficients satisfy $\alpha_r-\alpha_s=s-r$ and
$\alpha_r+\alpha_s=q-r-s$.
Since $f_\lambda$ is even, the product-to-sum identities give
\[
\begin{aligned}
C[\Phi]_{rs}
&=\frac12\left(
\widehat{f_\lambda}(s-r)
+\widehat{f_\lambda}(q-r-s)\right),\\
S[\Phi]_{rs}
&=\frac12\left(
\widehat{f_\lambda}(s-r)
-\widehat{f_\lambda}(q-r-s)\right).
\end{aligned}
\]
By \eqref{S:eq:pairy}, these equal $\lambda$ times the target
moments whenever both arguments are nonzero.
The argument $s-r$ vanishes when $r=s$.
The argument $q-r-s$ vanishes only when $r=s=q/2$.
At these entries, the value $\widehat{f_\lambda}(0)=1$
gives the correction $(1-\lambda)(C^\circ,S^\circ)$.

Under Haar measure on $\cT$, the coordinates
$\vartheta_r$ with $r\ne q/2$ are independent and uniform.
For even $q$, the remaining coordinate is $\vartheta_{q/2}=0$.
Their moments are exactly $(C^\circ,S^\circ)$.
\end{proof}

\subsection{Uniform average mixing}

\begin{theorem}\label{S:thm:main}
Let $q$ be a power of two at least four, an odd prime,
or twice an odd prime.
Then $P_{q-1}$ admits uniform average mixing under a probability
measure with at most $|\cR|^2+1$ atoms.
It also admits uniform average mixing under an absolutely
continuous probability measure.
\end{theorem}

\begin{proof}
For $q=3$, we have $M(\pi/4)=\frac12J$ on $P_2$.
The uniform distribution on $[0,\pi]$ also gives uniform
average mixing.

Now let $q\ge4$.
Define
\[
\widetilde F(\vartheta)
=\left(
(\cos\vartheta_r\cos\vartheta_s)_{r,s},
(\sin\vartheta_r\sin\vartheta_s)_{r,s}
\right)
\]
on $\T^{\cR}$.
Then $F(t)=\widetilde F((t\theta_r)_r)$.
Put
\[
\cK=\conv F(\R),
\quad
\cK_{\cT}=\conv\widetilde F(\cT).
\]
The set $\cK_{\cT}$ is compact.
Since the flow is dense in $\cT$, the set $F(\R)$ is dense
in $\widetilde F(\cT)$.
Consequently
\[
\operatorname{cl}\cK=\cK_{\cT},
\quad
\relint\cK=\relint\cK_{\cT},
\]
where the second equality follows from
\cite[Theorem~6.3]{Rock70}.

Apply Lemma~\ref{S:lem:bary} to $\widetilde F$ and Haar measure
on $\cT$.
It gives
\[
(C^\circ,S^\circ)\in\relint\cK_{\cT}.
\]
Choose $1<\lambda<\frac{N}{2g}$.
Such a choice is possible by Proposition~\ref{S:prop:measure}.
Proposition~\ref{S:prop:phases} and Lemma~\ref{S:lem:bary} give
\[
\Gamma_\lambda
=(1-\lambda)(C^\circ,S^\circ)
+\lambda(C^\star,S^\star)
\in\cK_{\cT}.
\]
Rearranging yields
\[
(C^\star,S^\star)
=\left(1-\frac1\lambda\right)(C^\circ,S^\circ)
+\frac1\lambda\Gamma_\lambda.
\]
Since $0<1/\lambda<1$, Theorem~6.1 of \cite{Rock70} places
this pair in $\relint\cK_{\cT}$.
Hence $(C^\star,S^\star)\in\cK$.

The set $F(\R)$ lies in an affine space of dimension at most
$|\cR|^2$.
Carath\'eodory's theorem \cite[Theorem~17.1]{Rock70} therefore gives
\[
(C^\star,S^\star)=\sum_j\mu_jF(t_j)
\]
as a convex combination with at most $|\cR|^2+1$ terms.
The probability measure
\[
\mu=\sum_j\mu_j\delta_{t_j}
\]
then satisfies $\Mh[\mu]=\frac1NJ$
by Proposition~\ref{S:prop:target}.

For the absolutely continuous case, let $\cK_{\mathrm{ac}}$
be the set of vectors $\int F\,d\mu$ obtained from absolutely
continuous probability measures $\mu$.
This set is convex.
It is contained in $\cK$ by Lemma~\ref{S:lem:bary}.

Fix $t\in\R$.
The uniform distributions on $[t-1/j,t+1/j]$ give points of
$\cK_{\mathrm{ac}}$ converging to $F(t)$ as $j\to\infty$.
Thus the closed convex set $\operatorname{cl}\cK_{\mathrm{ac}}$
contains $\operatorname{cl}\cK$.
The reverse inclusion follows from $\cK_{\mathrm{ac}}\subseteq\cK$.
Therefore
\[
\operatorname{cl}\cK_{\mathrm{ac}}=\cK_{\cT}.
\]
By \cite[Theorem~6.3]{Rock70},
\[
\relint\cK_{\mathrm{ac}}=\relint\cK_{\cT}.
\]
This set contains $(C^\star,S^\star)$.
An absolutely continuous probability measure therefore realizes
the target moments.
Proposition~\ref{S:prop:target} gives uniform average mixing.
\end{proof}

The proof is nonconstructive.
It guarantees the existence of readout times and a density
without specifying them.
The measures may be discrete, as in the sampling distributions
of \cite[Definition~2]{BCM24}.
They may also be absolutely continuous.
The obstructions in Section~\ref{S:sec:collisions} apply to
arbitrary Borel probability measures.

\section{Relation to pretty good state transfer}\label{S:sec:pgst}

Let $\cT'\subseteq\T^N$ be the closure of
\[
\{(t\theta_a)_{1\le a\le N}:t\in\R\}.
\]
Define $\gamma^\ast\in\T^N$ by
\[
\gamma^\ast_a=\left(\frac q2-a\right)\pi.
\]
Since $\theta_{q-a}=-\theta_a$, every point of $\cT'$ satisfies
$\vartheta_{q-a}=-\vartheta_a$.
The point $\gamma^\ast$ satisfies the same identities.
Restriction to the coordinates in $\cR$ identifies $\cT'$
with $\cT$.
Under this coordinate restriction, $\gamma^\ast$ maps to
$\gamma(\pi)$ (Fig.~\ref{fig:torus}).

\begin{proposition}\label{S:prop:pgst}
\begin{enumerate}[label=(\alph*),leftmargin=1.8em,itemsep=1pt,topsep=2pt]
\item
PGST between the end vertices of $P_{q-1}$ occurs if and only if
$\gamma^\ast+c\mathbf1\in\cT'$ for some $c\in\T$.
It occurs at every modulus in Theorem~\ref{S:thm:main}.
At these moduli, suitable times satisfy $U(t_j)\to-i^q\Iop$.

\item
For every $q$, the coherences of every uniformly mixing measure
are invariant under the sign change
\[
y(a,b)\mapsto(-1)^{a-b}y(a,b).
\]
The phase law of Proposition~\ref{S:prop:phases} is invariant
under translation by $\gamma(\pi)$.

\item
If there is a parity collision at $q$, then there is no PGST
between the end vertices of $P_{q-1}$.
\end{enumerate}
\end{proposition}

\begin{proof}
(a)
PGST between the ends means
\[
U(t_j)e_1\to\lambda e_N
\]
for some times $t_j$ and some $\lambda$ with $|\lambda|=1$.
The eigenvectors satisfy
\[
(v_a)_N=(-1)^{a-1}(v_a)_1\ne0.
\]
Expansion in these eigenvectors therefore gives the equivalent
condition
\[
e^{it_j\theta_a}\to\lambda(-1)^{a-1}
\quad\text{for every }a.
\]
Since $e^{i\gamma^\ast_a}=i^q(-1)^a$, this is equivalent to
\[
(t_j\theta_a)_a\to\gamma^\ast+c\mathbf1,
\quad
e^{ic}=-i^{-q}\lambda.
\]
This proves the first assertion.

At the moduli of Theorem~\ref{S:thm:main},
Lemma~\ref{S:lem:kronecker} gives $\gamma^\ast\in\cT'$.
We may therefore take $c=0$.
Lemma~\ref{S:lem:alt} then gives $U(t_j)\to-i^q\Iop$.

(b)
Translation by $\gamma^\ast$ multiplies
$\cos(\vartheta_a-\vartheta_b)$ by $(-1)^{a-b}$.
The target $y^\star$ is fixed by this sign change because
$\varkappa$ vanishes at odd arguments.
The first assertion follows from Theorem~\ref{S:thm:forced}.

For the second assertion, the density $f_\lambda$ has period $\pi$.
Also,
\[
\gamma(\xi+\pi)=\gamma(\xi)+\gamma(\pi).
\]
Thus the phase law is invariant under translation by $\gamma(\pi)$.

(c)
A parity collision $(k,l),(k',l')$ gives the integer vector
\[
\boldsymbol{\ell}
=\mathbf{e}_k-\mathbf{e}_l-\mathbf{e}_{k'}
+\mathbf{e}_{l'}\in\Z^N.
\]
It satisfies
\[
\sum_a\ell_a\theta_a=0,
\quad
\sum_a\ell_a=0.
\]
On the other hand,
\[
\sum_a\ell_a\gamma^\ast_a
=\bigl((l-k)-(l'-k')\bigr)\pi
\equiv\pi\pmod{2\pi}.
\]
Every point of $\cT'$ satisfies
$\sum_a\ell_a\vartheta_a=0$ in $\T$.
For every $c\in\T$, however,
\[
\sum_a\ell_a(\gamma^\ast_a+c)\equiv\pi\pmod{2\pi}.
\]
Hence $\gamma^\ast+c\mathbf1\notin\cT'$ for every $c$.
Part (a) excludes PGST between the ends.
\end{proof}

The uniformizing phase law is therefore built on the transfer
symmetry.
The same integer relation that makes the forced coherences
contradictory also excludes end-to-end transfer.
Part (c) also follows from \cite{GKSS12}.
Indeed, no modulus in Theorem~\ref{S:thm:collset} is a power
of two, a prime, or twice a prime.

This connection concerns only the end vertices.
Van Bommel \cite[Theorem~4]{vB19} characterized PGST between
arbitrary vertices of paths.
The sufficiency result for internal vertices is due to Coutinho,
Guo and van Bommel \cite[Theorem~2]{CGvB17}.
For example, $P_{11}$ has PGST between vertices $2$ and $10$.
Here $q=12$ is a collision modulus, so $P_{11}$ admits no
uniform average mixing.

\section{Beyond chains}\label{S:sec:beyond}

\subsection{Cartesian products}

\begin{theorem}\label{S:thm:product}
Let $G$ and $H$ be graphs on $n$ and $m$ vertices, respectively.
Let $\mu$ be a Borel probability measure.
Then $\Mh^{G\square H}[\mu]=\frac1{nm}J$ if and only if
\begin{equation}\label{S:eq:uncorrelated}
\int_\R M_G(t)_{uv}\,M_H(t)_{u'v'}\,d\mu(t)=\frac1{nm}
\end{equation}
for all $u,v\in V(G)$ and $u',v'\in V(H)$.
This condition implies
\[
\Mh^G[\mu]=\frac1nJ,
\quad
\Mh^H[\mu]=\frac1mJ.
\]
In particular, $G\square G$ admits uniform average mixing
if and only if $G$ admits instantaneous uniform mixing.
\end{theorem}

\begin{proof}
The identity
\[
A(G\square H)=A(G)\otimes I+I\otimes A(H)
\]
has commuting summands.
Hence
\[
U_{G\square H}(t)=U_G(t)\otimes U_H(t).
\]
It follows that
\[
M_{G\square H}(t)_{(u,u'),(v,v')}
=M_G(t)_{uv}M_H(t)_{u'v'}.
\]
Integration gives the equivalence.

Sum \eqref{S:eq:uncorrelated} over $v'$.
Since each row of $M_H(t)$ sums to $1$, we obtain
\[
\int_\R M_G(t)_{uv}\,d\mu(t)=\frac1n.
\]
The same argument gives the assertion for $H$.

Now suppose that $H=G$ and that $G\square G$ admits uniform
average mixing under $\mu$.
Taking $(u',v')=(u,v)$ gives
\[
\int_\R\left(M_G(t)_{uv}-\frac1n\right)^2d\mu(t)
=\frac1{n^2}-\frac2{n^2}+\frac1{n^2}
=0.
\]
Thus $M_G(t)_{uv}=1/n$ for $\mu$-almost every $t$.
There are only finitely many entries.
Consequently $M_G(t)=\frac1nJ$ for $\mu$-almost every $t$,
so $G$ admits instantaneous uniform mixing.

Conversely, suppose that $M_G(t_0)=\frac1nJ$.
Then $\mu=\delta_{t_0}$ gives uniform average mixing
on $G\square G$.
\end{proof}

\subsection{Trees with instantaneous uniform mixing}

\begin{lemma}[{\cite{GZ17}}]\label{S:lem:star}
The star $K_{1,3}$ admits instantaneous uniform mixing.
More precisely, $M(t_0)=\frac14J$ whenever
$\cos(\sqrt3\,t_0)=-\frac12$.
\end{lemma}

\begin{proof}
Let $o$ be the center, and let $\ell\ne\ell'$ be leaves.
Put $c=\cos(\sqrt3t)$.
The spectral decomposition gives
\[
\begin{aligned}
U_{oo}&=c,
&
U_{o\ell}&=\frac{i}{\sqrt3}\sin(\sqrt3t),\\
U_{\ell\ell}&=\frac{c+2}{3},
&
U_{\ell\ell'}&=\frac{c-1}{3}.
\end{aligned}
\]
At $c=-1/2$, all four squared moduli equal $1/4$.
\end{proof}

\begin{lemma}[Degree parity]\label{S:lem:parity}
Let $G$ be a bipartite graph on $n$ vertices with adjacency matrix $A$.
Let $\Sigma$ be as in Lemma~\ref{S:lem:bip}.
Suppose that $G$ admits $\epsilon$-uniform mixing.

Then the closure of $\{U(t):t\in\R\}$ contains a matrix $Q$
such that $|Q_{uv}|^2=1/n$ for all $u,v$.
If instantaneous uniform mixing occurs at $t_0$,
one may take $Q=U(t_0)$.
Every such $Q$ satisfies
\[
AQ=QA,
\quad
\Sigma Q\Sigma=\overline Q.
\]
Moreover, all vertices of $G$ have degrees of the same parity.
\end{lemma}

\begin{proof}
Choose times $t_j$ such that $M(t_j)\to\frac1nJ$.
By compactness of the unitary group, a subsequence of $U(t_j)$
converges to a matrix $Q$.
Its entries satisfy $|Q_{uv}|^2=1/n$.

Each $U(t)$ commutes with $A$ and satisfies
\[
\Sigma U(t)\Sigma=e^{-itA}=\overline{U(t)}.
\]
These identities pass to the limit.

Put $\widehat Q=\sqrt n\,Q$.
The identity $\Sigma Q\Sigma=\overline Q$ gives
\[
\widehat Q_{uv}\in
\begin{cases}
\{1,-1\},&\sigma_u\sigma_v=1,\\
\{i,-i\},&\sigma_u\sigma_v=-1.
\end{cases}
\]
Fix vertices $a,b$.
Comparing the $(a,b)$ entries of $A\widehat Q=\widehat QA$ gives
\[
\sum_{z\sim a}\widehat Q_{zb}
=\sum_{y\sim b}\widehat Q_{ay}.
\]
Adjacent vertices have opposite signs.
Thus every term on either side belongs to $\{\omega,-\omega\}$,
where
\[
\omega=
\begin{cases}
1,&\sigma_a\sigma_b=-1,\\
i,&\sigma_a\sigma_b=1.
\end{cases}
\]
After division by $\omega$, the left side is a sum of
$\deg a$ numbers in $\{1,-1\}$.
The right side is a sum of $\deg b$ such numbers.
Reducing modulo $2$ yields $\deg a\equiv\deg b\pmod2$.
\end{proof}

For instantaneous uniform mixing, the entry structure of
$\widehat Q$ is given in \cite[Lemma~6.1]{GMR17}.
The limit argument for $\epsilon$-uniform mixing follows
\cite[Proposition~2]{Mon26}.

\begin{lemma}[{\cite[Proposition~8]{Mon26}}]\label{S:lem:cherry}
Let $T$ be a tree on $n\ge3$ vertices with no vertex of degree two.
Then two leaves of $T$ have a common neighbor.
\end{lemma}

\begin{proof}
Let $\mathcal L$ be the set of leaves, and let $\mathcal N$
be the set of remaining vertices.
Every vertex in $\mathcal N$ has degree at least $3$.
Since $n\ge3$, the neighbor of each leaf belongs to $\mathcal N$.

The identity $\sum_v(\deg v-2)=-2$ gives
\[
|\mathcal L|
=2+\sum_{v\in\mathcal N}(\deg v-2)
\ge2+|\mathcal N|.
\]
Thus the map sending each leaf to its neighbor cannot be injective.
\end{proof}

\begin{theorem}\label{S:thm:tree}
A tree admits instantaneous uniform mixing if and only if
it is $K_1$, $P_2$ or $K_{1,3}$.
The same classification holds for $\epsilon$-uniform mixing.
\end{theorem}

\begin{proof}
The case $K_1$ is trivial.
The path $P_2$ mixes uniformly at $t=\pi/4$.
The star $K_{1,3}$ does so by Lemma~\ref{S:lem:star}.
Instantaneous uniform mixing also implies $\epsilon$-uniform mixing.

Conversely, let $T$ be a tree on $n\ge2$ vertices with
$\epsilon$-uniform mixing.
Choose $Q$ as in Lemma~\ref{S:lem:parity}.
Since $T$ has a leaf, all its degrees are odd.

If $n=2$, then $T=P_2$.
Suppose that $n\ge3$.
Lemma~\ref{S:lem:cherry} gives two leaves $u\ne w$
with a common neighbor.
The vector $x=e_u-e_w$ satisfies $Ax=0$.
Hence $U(t)x=x$ for every $t$, and passing to the limit gives $Qx=x$.
The $u$th coordinate is
\[
Q_{uu}-Q_{uw}=1.
\]
Since $|Q_{uu}|=|Q_{uw}|=n^{-1/2}$, we obtain
\[
1\le |Q_{uu}|+|Q_{uw}|=\frac2{\sqrt n}.
\]
Thus $n\le4$.
Among trees on three or four vertices, only $K_{1,3}$
has all degrees odd.
\end{proof}

\begin{corollary}\label{S:cor:squares}
For a tree $T$, the Cartesian square $T\square T$ admits
uniform average mixing if and only if
$T\in\{K_1,P_2,K_{1,3}\}$.
In particular, for every prime $p\ge5$, the square array
$P_{p-1}\square P_{p-1}$ admits no uniform average mixing.
The path $P_{p-1}$ itself admits uniform average mixing
by Theorem~\ref{S:thm:main}.
\end{corollary}

\begin{proof}
Apply Theorems~\ref{S:thm:product} and~\ref{S:thm:tree}.
\end{proof}

\end{document}